\pdfoutput=1
\RequirePackage{ifpdf}
\ifpdf 
\documentclass[pdftex]{sigma}
\else
\documentclass{sigma}
\fi

\numberwithin{equation}{section}

\newtheorem{Theorem}{Theorem}[section]
\newtheorem{Corollary}[Theorem]{Corollary}
\newtheorem{Proposition}[Theorem]{Proposition}
{ \theoremstyle{definition}
\newtheorem{Definition}[Theorem]{Definition}

\newtheorem{Remark}[Theorem]{Remark} }

\begin{document}

\allowdisplaybreaks

\newcommand{\arXivNumber}{1803.01247}

\renewcommand{\PaperNumber}{099}

\FirstPageHeading

\ShortArticleName{Generalized Lennard-Jones Potentials, SUSYQM and Differential Galois Theory}

\ArticleName{Generalized Lennard-Jones Potentials,\\ SUSYQM and Differential Galois Theory}

\Author{Manuel F. ACOSTA-HUM\'ANEZ~$^{\dag^1}$, Primitivo B. ACOSTA-HUM\'ANEZ~$^{\dag^2\dag^3}$\newline and Erick TUIR\'AN~$^{\dag^4}$}

\AuthorNameForHeading{M.F.~Acosta-Hum\'anez, P.B.~Acosta-Hum\'anez and E.~Tuir\'an}

\Address{$^{\dag^1}$~Departamento de F\'{\i}sica, Universidad Nacional de Colombia,\\
\hphantom{$^{\dag^1}$}~Sede Bogot\'{a}, Ciudad Universitaria 111321, Bogot\'a, Colombia}
\EmailDD{\href{mailto:mafacostahu@unal.edu.co}{mafacostahu@unal.edu.co}}

\Address{$^{\dag^2}$~Facultad de Ciencias B\'asicas y Biom\'edicas, Universidad Sim\'on Bol\'{\i}var,\\
\hphantom{$^{\dag^2}$}~Sede 3, Carrera 59 No.~58--135. Barranquilla, Colombia}
\EmailDD{\href{mailto:primitivo.acosta@unisimonbolivar.edu.co}{primitivo.acosta@unisimonbolivar.edu.co}}
\URLaddressDD{\url{http://www.intelectual.co/primi/}}

\Address{$^{\dag^3}$~Instituto Superior de Formaci\'on Docente Salom\'e Ure\~na - ISFODOSU,\\
\hphantom{$^{\dag^3}$}~Recinto Emilio Prud'Homme, Calle R. C. Tolentino \#51, esquina 16 de Agosto,\\
\hphantom{$^{\dag^3}$}~Los Pepines, Santiago de los Caballeros, Rep\'ublica Dominicana}

\Address{$^{\dag^4}$~Departamento de F\'{\i}sica y Geociencias, Universidad del Norte,\\
\hphantom{$^{\dag^4}$}~Km 5 V\'{\i}a a Puerto Colombia AA 1569, Barranquilla, Colombia}
\EmailDD{\href{mailto:etuiran@uninorte.edu.co}{etuiran@uninorte.edu.co}}

\ArticleDates{Received May 01, 2018, in final form September 14, 2018; Published online September 19, 2018}

\Abstract{In this paper we start with proving that the Schr\"{o}dinger equation (SE) with the classical $12-6$ Lennard-Jones (L-J) potential is nonintegrable in the sense of the differential Galois theory (DGT), for any value of energy; i.e., there are no solutions in closed form for such differential equation. We study the $10-6$ potential through DGT and SUSYQM; being it one of the two partner potentials built with a~superpotential of the form $w(r)\propto 1/r^5$. We also find that it is integrable in the sense of DGT for zero energy. A~first analysis of the applicability and physical consequences of the model is carried out in terms of the so called
De Boer principle of corresponding states. A comparison of the second virial coefficient~$B(T)$ for both potentials shows a good agreement for low temperatures. As a consequence of these results we propose the $10-6$ potential as an integrable alternative to be applied in further studies instead of the original $12-6$ L-J potential. Finally we study through DGT and SUSYQM the integrability of the SE with a generalized $(2\nu-2)-\nu$ L-J potential. This analysis do not include the study of square integrable wave functions, excited states and energies different than zero for the generalization of L-J potentials.}

\Keywords{Lennard-Jones potential; differential Galois theory; SUSYQM; De Boer principle of corresponding states}

\Classification{12H05; 81V55; 81Q05}

\section{Introduction}
The Lennard-Jones potential (L-J) was proposed in 1931 in order to model the concurrence between the long-range attraction and the short-range repulsion in radial interatomic interactions~\cite{LJ}. In a later work, the
description of such potential was employed in order to describe the equation of state of a gas in terms of its interatomic forces~\cite{LJ2}, thus concluding and enhancing an investigation started by Mie in 1903~\cite{Mie}. The L-J potential is usually used, at the level of classical statistical mechanics, to study the behavior of fluid materials, ranging from simple molecules to polymers and proteins~\cite{HansenVerlet, Binder, Mecke}. In theoretical quantum chemistry, among many applications, we point out: its implementation in the theory of molecular orbitals, allowing to compute the tendency of two electrons in the same space orbital to keep each other apart because of the repulsive field between them~\cite{Hurley}; the numerical implementations in order to compute the transferable inter-molecular potential functions (TIPS) in alcohols, ethers and water, that have given an understanding of the interactions of these chemical compounds in solvents~\cite{Jorgensen}. Also a mathematical model that has been proposed for calculating the isosteric heat of adsorption of simple fluids onto flat surfaces. On this respect, theoretical and experimental results were compared in order to study the influence of the choice of the intermolecular potential parameters~\cite{Muleros}. Finally, a experimental methodology and theoretical calculations applying the Lennard-Jones potential, for determining micropore-size distributions, obtained from physical adsorption isotherm data, have provided valuable microstructural information, which is still widely used today~\cite{HK, Olivier, Storck}.

With the increase of numerical techniques, calculations with explicit solutions in physical models don't have in the present the same importance as in past decades. Nevertheless exact solutions when available, have always served as elucidating tools for finding general properties of the system, which otherwise could remain hidden. The main motivation of this paper is the application of supersymmetric quantum mechanics (SUSYQM) and differential Galois theory (DGT) to obtain explicit solutions of the Schr\"{o}dinger equation (SE) with variants of the Lennard-Jones potential, as well as the set of eigenvalues associated to each solution.

SUSYQM, introduced by E.~Witten in 1981, is the simplest example where supersymmetry can be dynamically broken~\cite{witten}. In spite of its initial character of a toy model; SUSYQM has earned importance in the recent decades, because it served as a starting point to the development of attractive theoretical features and concepts like shape invariance, isospectrality and factorization, that give new perspectives to old problems in quantum mechanics, like the integrability of the SE, see for example~\cite{acthesis,Gango,cks} and the path integral formulation of classical mechanics~\cite{Reuter-Gozzi}. On the other hand, there is plenty of papers in
mathematical physics wherein DGT has been applied; see for example~\cite{ac,aabd, aad,acbl,acbl2,acblva} for applications to study the non-integrability of Hamiltonian systems. For applications in the integrability of the SE, see~\cite{acthesis,acbook,akmss,acmowe,acpan,acsu2,acsu}. For applications of differential Galois theory to other quantum integrable systems see~\cite{braverman,semenov}. The main Galoisian tools used in some of these papers are the Hamiltonian algebrization and the Kovacic's algorithm. These tools have led and still lead, to deduce exact solutions in several areas of mathematical physics.

The structure of this paper is as follows. Section~\ref{section2} is devoted to the theoretical background necessary to understand the rest of the paper. It summarizes topics such as the Schr\"{o}dinger equation for central potentials, Lennard-Jones potentials $12-6$, $10-6$ and $(2\nu-2)-\nu$, SUSYQM, the De Boer principle of corresponding states, the virial equation and DGT. In Section~\ref{section3} we study the integrability of the SE with the usual $12-6$ Lennard-Jones potential, as well as the alternative versions $10-6$ and $(2\nu-2)-\nu$. Our contributions consist in the deduction of algebraic and physical conditions over the parameters of such SE's to get their integrability in the sense of DGT and the superpotentials in the integrable cases. A first study of physical consequences will also be detailed in this section. In Section~\ref{section4} some remarks concerning future works are established.

\section{Preliminaries}\label{section2}
\subsection*{The Schr\"{o}dinger equation for a central potential}

We are interested in studying a physical model for a many-body system where the main contribution of the interaction of its constituents is pairwise and radial in nature. In addition, the physical conditions of the system (temperature, density, etc) are such, that its quantum behavior is non-negligible. In this section we set shortly the theoretical background, in order to establish our physical model with a central potential, and also the notation to be applied in the rest of the paper~\cite{Cohen-T1}. The Hamiltonian for a system of two spinless particles with masses~$m_{1}$ and~$m_{2}$ interacting via a radial potential $V\left( \left\vert \vec{r}_{1}-\vec{r}_{2}\right\vert \right) $ is given by
\begin{gather}
H=T+V=\frac{p_{1}^{2}}{2m_{1}}+\frac{p_{2}^{2}}{2m_{2}}+V(\vert \vec{r}_{1}-\vec{r}_{2}\vert). \label{dos parts}
\end{gather}
It is an usual subject of textbooks in classical mechanics to show that~(\ref{dos parts}) can be separated into two parts, one related to the motion of the center of mass $\vec{R}$ of the system and the other related to the
relative motion of the particles. The new coordinate system is given by the following transformation rules
\begin{gather*}
\vec{R}=\frac{m_{1}\vec{r}_{1}+m_{2}\vec{r}_{2}}{m_{1}+m_{2}} ,\qquad \vec{r}=\vec{r}_{1}-\vec{r}_{2} ,\qquad \mu =\frac{m_{1}m_{2}}{m_{1}+m_{2}} , \qquad M=m_{1}+m_{2} , \nonumber \\ \vec{p}_{r}=\mu \frac{{\rm d}\vec{r}}{{\rm d}t} ,\qquad \vec{p}_{R}=M\frac{{\rm d}\vec{R}}{{\rm d}t}, 
\end{gather*}
where $M$ is the total mass of the system, $\mu $ is called the reduced mass. The Hamiltonian in the new coordinates takes the form
\begin{gather}
H=\frac{p_{r}^{2}}{2\mu }+\frac{p_{R}^{2}}{2M}+V(r),\label{Con CM}
\end{gather}
where $p_{r}$ and $p_{R}$ are the canonical momenta conjugated respectively to the coordinates $r= \vert \vec{r}_{1}-\vec{r}_{2} \vert $ and $R= \vert \vec{R} \vert $. Since we are not dealing with external forces, the motion of the center of mass is uniform rectilinear. For several analysis it is suitable to work in a frame at rest with the center of mass, which is still an inertial reference frame, in that case the Hamiltonian~(\ref{Con CM}) is reduced to
\begin{gather}
H=\frac{p_{r}^{2}}{2\mu }+V(r). \label{el que es}
\end{gather}
The Hamiltonian in (\ref{el que es}) represents the energy of the relative motion of the two particles; it describes the motion of a fictitious particle, the \emph{relative particle} with a mass given by the reduced mass $\mu $ and a position and momentum given by the relative coordinates~$\vec{r}$ and~$\vec{p}_{r}$. The quantum mechanical model of our interest is based on this Hamiltonian. The usual rules of quantization in the position
representation lead to the time-independent Schr\"{o}dinger equation for our two-particles system
\begin{gather}
\left[ -\left( \frac{\hbar ^{2}}{2\mu }\right) \vec{\nabla}^{2}+V(r) \right] \Psi ( \vec{r}) =E\Psi ( \vec{r}).\label{Ec schroed 1}
\end{gather}
Since $V(r) $ is a rational central potential, the eigenfunctions $\Psi( \vec{r}) $ are separable into radial and angular parts, the last one given by the spherical harmonics
\begin{gather*}
\Psi ( \vec{r}) =\frac{1}{r}u_{k,l}(r) Y_{l}^{m}( \theta ,\varphi ). 
\end{gather*}
The differential equation of our interest corresponds to the radial part of~(\ref{Ec schroed 1}) as follows
\begin{gather*}
\left[ -\left( \frac{\hbar ^{2}}{2\mu }\right) \frac{{\rm d}^{2}}{{\rm d}r^{2}}+\frac{l(l+1) }{2\mu r^{2}}+V(r) \right] u_{k,l}(r) =E_{k,l}u_{k,l}(r) , 
\end{gather*}
where $l$ and $m$ are the usual quantum numbers for angular momentum; $k$ represents the different values of energy for fixed~$l$, and it can be either discrete or continuous. Defining the \emph{effective radial potential} as $V^{\rm eff}(r) \equiv l(l+1) /\big( 2\mu r^{2}\big) +V(r)$ and leaving the second derivative in $r$ on one side, we have
\begin{gather}
\left( \frac{2\mu }{\hbar ^{2}}\right) \big[ V^{\rm eff}(r) -E_{k,l}\big] u_{k,l}(r) =\frac{{\rm d}^{2}}{{\rm d}r^{2}}u_{k,l}(r)\label{Schroed com hbar}
\end{gather}
at this point we define a rescaled potential $v^{\rm eff}(r)\equiv \big( \frac{2\mu }{\hbar ^{2}}\big) V^{\rm eff}(r)\ $and a similarly rescaled energy $\big( \frac{2\mu }{\hbar ^{2}}\big) E_{k,l}\equiv \varepsilon _{k,l}$; in this case equation~(\ref{Schroed com hbar}) turns out to be
\begin{gather}
\big[ v^{\rm eff}(r)-\varepsilon _{k,l}\big] u_{k,l}(r) =\frac{{\rm d}^{2}}{{\rm d}r^{2}}u_{k,l}(r). \label{Rad Schroed scaled}
\end{gather}
In this way it is natural to define a~rescaled Hamiltonian as $\mathcal{H}\equiv \big[{-}\frac{{\rm d}^{2}}{{\rm d}r^{2}}+v(r)\big] $ in order to recover~(\ref{Rad Schroed scaled}):
\begin{gather}
\mathcal{H}_{\rm eff}u_{k,l}(r) \equiv \left[ -\frac{{\rm d}^{2}}{{\rm d}r^{2}}+v^{\rm eff}(r)\right] u_{k,l}(r) =\varepsilon_{k,l}u_{k,l}(r). \label{Con H}
\end{gather}
 The case for $l=0$ defines the non-effective potential, and is also of great interest for our study
\begin{gather}
\mathcal{H}u_{k}(r) \equiv \left[ -\frac{{\rm d}^{2}}{{\rm d}r^{2}}+v(r)\right] u_{k}(r) =\varepsilon_{k}u_{k}(r), \label{Con H l0}
\end{gather}
where we have simplified $u_{k,l=0}$ and $\varepsilon_{k,l=0}$ to $u_{k}$ and $\varepsilon_{k}$, respectively. We observe that (\ref{Con H l0}) is a~rescaled version of
\begin{gather}
\widetilde{\mathcal{H}}u_{k}(r) \equiv \left[ -\left(\frac{\hbar^{2}}{2\mu}\right)\frac{{\rm d}^{2}}{{\rm d}r^{2}}+V(r)\right] u_{k}(r) =E_{k}u_{k}(r). \label{Hl0 Non}
\end{gather}
equations (\ref{Con H}), (\ref{Con H l0}) and (\ref{Hl0 Non}) are the subject of our mathematical and physical analysis in Section~\ref{section3}.

\subsection*{The $\boldsymbol{12-6}$ Lennard-Jones potential and its generalizations}

The 12-6 Lennard-Jones potential is usually presented in terms of two constants $A$ and $B$
\begin{gather}
V_{12-6}(r)=-{\frac{A}{r^{6}}}+{\frac{B}{r^{12}}}, \label{LJ 12-6}
\end{gather}
where the negative term $-A/r^{6}$ leads to van der Waals attractive fields and comes from the second-order correction in perturbation theory to the dipole-dipole interaction between two atoms~\cite{Cohen-T1}. The positive
term $B/r^{12}$ models the short range electronic repulsion between atoms and has no theoretical justification; it was empirically chosen because it fits reasonably good data coming from experiments with diatomic gases \cite{LJ}. An alternative version is given by
\begin{gather}\label{12-6 Sigma-Eps}
V_{12-6}(r) = 4\epsilon\left[ \left( \frac{\sigma}{r}\right)^{12}-\left( \frac{\sigma}{r}\right) ^{6}\right], \qquad A =4\epsilon \sigma ^{6} ,\qquad B=4\epsilon \sigma ^{12},
\end{gather}
where $\epsilon $ is the atomic depth of the potential well, $\sigma $ is the finite distance at which the inter-particle potential is zero and $r$ is the distance between the particles (see Fig.~\ref{Fig1}). In mathematical terms, $\sigma >0 $ and $\epsilon >0$ satisfy that $V(\sigma )=0$ and $V\big(\sqrt[6]{2}\sigma\big)=-\epsilon $; this means that $\sigma $ is a zero potential length and the point $\big(
\begin{matrix} \sqrt[6]{2}\sigma , & -\epsilon \end{matrix} \big) $ is the local minimum of the potential in the inter\-val~$(0,\infty )$. It can easily be shown that there is no other critical point in such interval. In physical terms the well depth $\epsilon $ and the zero potential length $\sigma $ are parameters that describe the cohesive and repulsive forces that take place in a gas or liquid at the molecular level. $\epsilon $ measures the strength of the attraction between pairs of molecules and~$\sigma $ is the radius of the repulsive core when two molecules collide.
\begin{figure}[t]\centering
\includegraphics[width=80mm]{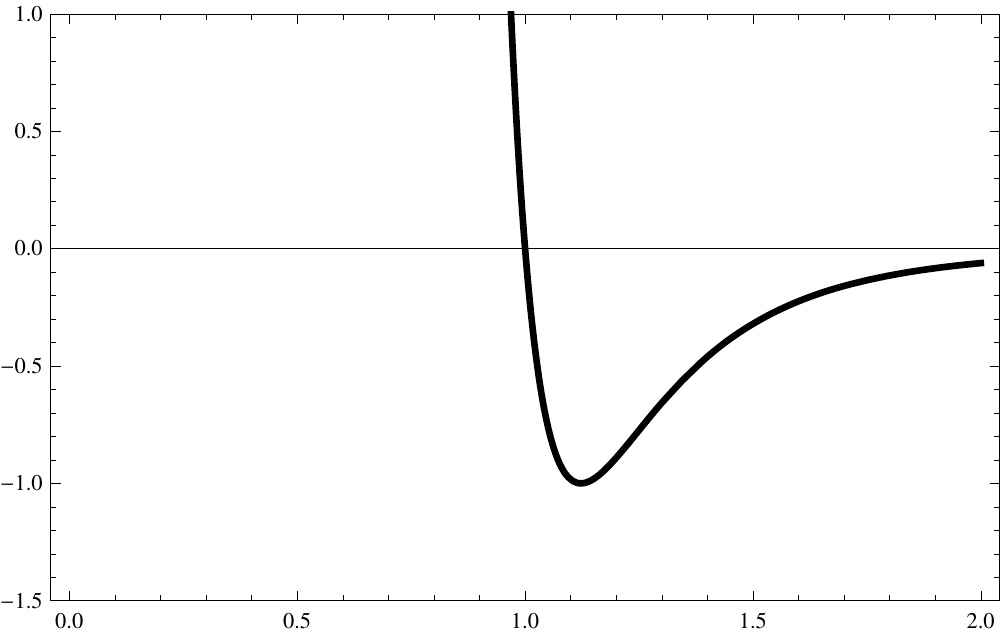}
\caption{$\frac{v(r)}{\epsilon}$ vs.\ $\frac{r}{\sigma}$ plot for the rescaled $12-6$ Lennard-Jones potential given in equation~\eqref{Gral LJ Resc} for $C=0$.}\label{Fig1}
\end{figure}

In order to explore with the differential Galois theory the integrability of the Schr\"{o}dinger equation with the Lennard-Jones potential (\ref{LJ 12-6}) and other related cases, we introduce the generalized effective version with arbitrary powers $\nu $ and $\delta$ given in~\cite{LJ}
\begin{gather*}
 V_{\delta -\nu }(r) \equiv -{\frac{A}{r^{\nu }}}+{\frac{B}{r^{\delta }}}, \qquad
 V_{\delta -\nu }^{\rm eff}(r) \equiv V_{\delta -\nu }(r)+\frac{C}{r^{2}}=-{\frac{A}{r^{\nu }}}+{\frac{B}{r^{\delta}}}+\frac{C}{r^{2}}, 
\end{gather*}
where $0<\nu< \delta$, $A>0$, $B>0$, $C\geq 0$. Its rescaled version is given by
\begin{gather}
v_{\delta -\nu }(r) \equiv \left( \frac{2\mu }{\hbar ^{2}}\right)V_{\delta -\nu }(r)= -{\frac{\bar{A}}{r^{\nu }}}+{\frac{\bar{B}}{r^{\delta }}},\label{LJ Resc 1} \\
v_{\delta -\nu }^{\rm eff}(r) \equiv \left( \frac{2\mu }{\hbar ^{2}}\right)V_{\delta -\nu }^{\rm eff}(r)=-{\frac{\bar{A}}{r^{\nu }}}+{\frac{\bar{B}}{r^{\delta }}}+\frac{\bar{C}}{r^{2}}, \label{Gral LJ Resc} \\
\bar{A}\equiv \left( \frac{2\mu }{\hbar ^{2}}\right) A,\qquad \bar{B}\equiv \left( \frac{2\mu }{\hbar ^{2}}\right) B,\qquad \bar{C}\equiv \left( \frac{2\mu }{\hbar ^{2}}\right) C. \nonumber
\end{gather}
The special case for $\delta=2\nu-2$ and some of its analytic advantages has been studied by J.~Pade in~\cite{Pade}
\begin{gather*}
 V_{(2\nu-2)-\nu }(r) \equiv -{\frac{A}{r^{\nu }}}+{\frac{B}{r^{2\nu-2}}}. 
\end{gather*}
 In the mentioned article, a special attention has been drawn to the $\nu=6$ case, and its ability to fit experimental data:
\begin{gather}
 V_{10-6 }(r) \equiv -{\frac{A}{r^{6 }}}+{\frac{B}{r^{10}}}. \label{LJ pot Diez-Seis}
\end{gather}
In Section~\ref{section3} we will explore the interesting features of (\ref{LJ pot Diez-Seis}) in the realm of SUSYQM.

\subsection*{The second virial coefficient and its dependence on the potential}
The virial equation of state for a gas expresses the deviation from the ideal behavior as a power series in the density $\rho$:
\begin{gather*}
\frac{p}{kT}=\rho+B_{2}(T)\rho^{2}+B_{3}(T)\rho^{3}+\cdots.
\end{gather*}
The coefficients $B_{n}(T)$ are called the virial coefficients and they are unique real functions of the temperature. The second virial coefficient~$B_2(T)$ represents the most significant deviation from the ideal behavior, since it is the prefactor in the term of order~$\rho^{2}$ in the series. It is a~customary result from equilibrium statistical mechanics (see~\cite{Quarrie}) that~$B_2(T)$ is a radial integral of the pair-potential $v(r)$ given by
\begin{gather}
B_2(T)=2\pi\int^{\infty}_0\big(1-{\rm e}^{-\frac{v(r)}{kT}}\big)r^2{\rm d}r. \label{Virial B}
\end{gather}
A thorough study by Keller and Zumino of the properties of~(\ref{Virial B}) has shown that a unique potential function can only be obtained from $B_2(T)$ if the potential behaves monotonically~\cite{Zumino}. This is clearly not the case for the Lennard-Jones potential and all its variants. As a~result, there exists an ambiguity in the choice of the microscopic potential, leading to the same thermodynamic function~$B_2(T)$. In addition to this analytic inexactness there is also the limited range of measurements of~$B_2(T)$ for low temperatures. The aforementioned limitations lead to several possibilities of choice for $v(r)$, at least from measurements of $B_2$, specially for the power of the repulsive term~$B/r^{ \delta}$. The possibilities range from $n=9$ to $n=14$ since the early works of Lennard-Jones (see~\cite{L-J1924I,L-J1924II}) and De Boer (see~\cite{DeBoer38}). We come back to this point in the next section, giving some hints about the applicability of the $10-6$ potential for low temperatures.

\subsection*{The dimensionless Schr\"{o}dinger equation\\ and the De Boer principle of corresponding states}

In 1948 J.~De Boer introduced a dimensionless representation of the Schr\"{o}dinger equation employing $\sigma $ and $\epsilon $ in order to construct dimensionless lengths and energies~\cite{DeBoer48}
\begin{gather}
\tilde{r}\equiv \frac{r}{\sigma},\qquad \tilde{E}\equiv \frac{E}{\epsilon} ,\qquad \tilde{V}\equiv \frac{V}{\epsilon}. \label{De Boer Trafo}
\end{gather}
As a result, the radial Schr\"{o}dinger equation (\ref{Hl0 Non}) for $l=0$ can be transformed into the dimensionless form given by
\begin{gather}
\left[ -\frac{\Lambda ^{2}}{2}\frac{{\rm d}^{2}}{{\rm d}\tilde{r}^{2}}+\tilde{V}(\tilde{r})\right] u(\tilde{r}) =\tilde{E}u(\tilde{r}) \label{Dim Less}
\end{gather}
provided that the potential $V(r)$ can be expressed in the generic form $V(r)=\epsilon f(r/\sigma)$, where~$f(r)$ is a well-defined dimensionless interaction function and $\Lambda \equiv \hbar /(\sigma \sqrt{\mu \epsilon })$~\cite{DeBoer48}. From~(\ref{Dim Less}) we see that $\Lambda$, the so-called De Boer parameter, is the only parameter in the equation that gives information about the particular microscopic characteristics of the system. From this fact, De Boer was able to formulate his principle of corresponding states, which is a ``quantum'' generalization of the van der Waals law of corresponding states for classical gases and liquids~\cite{Guggenheim,Pitzer, VdW}. The De Boer principle of corresponding states tells us that two different systems with equal value of $\Lambda $ have identical thermodynamic properties~\cite{DeBoer48}. In Section~\ref{section3} we exploit this principle in order to give an interpretation to the supersymmetric integrable model for zero energy, we propose with the $10-6$ Lennard-Jones potential.
\subsection*{Supersymmetric quantum mechanics}
We implement in this work the simplest realisation of SUSYQM for one-dimensional quantum systems~\cite{Gango}, which includes besides the Hamiltonian operator $H$, two \emph{fermionic} operators $Q^{\pm }$ or \emph{supercharges} such that they commute with~$H$
\begin{gather}
\big[ Q^{\pm },H\big] =0 \label{Qmas_men}
\end{gather}%
and satisfy the algebra
\begin{gather}
\big\{ Q^{-},Q^{+}\big\} =H ,\qquad \big( Q^{\pm }\big) ^{2}=0.\label{Alg_Q_mas_men}
\end{gather}
The second relation means that $Q^{\pm }$ are nilpotent operators. A usual representation of the algebra, given in equations~\eqref{Qmas_men} and~\eqref{Alg_Q_mas_men}, presents the Hamiltonian~$H$ of the system, as a~diagonal two component matrix of partner Hamiltonians~$H_{\pm }$
\begin{gather*}
H\equiv \left(
\begin{matrix}
H_{+} & 0 \\
0 & H_{-}
\end{matrix}
\right), 
\end{gather*}
where $Q^{\pm }$ are $2\times 2$ diagonal matrices involving the \emph{Ladder} operators $A^{\pm }$
\begin{gather*}
Q^{-}\equiv \left(
\begin{matrix}
0 & 0 \\
A^{-} & 0
\end{matrix}
\right) ,\qquad Q^{+}\equiv \left(
\begin{matrix}
0 & A^{+} \\
0 & 0
\end{matrix}
\right), 
\end{gather*}
such that%
\begin{gather}
H\equiv \big\{ Q^{-},Q^{+}\big\} \equiv Q^{-}Q^{+}+Q^{+}Q^{-}=\left(
\begin{matrix}
A^{+}A^{-} & 0 \\
0 & A^{-}A^{+}%
\end{matrix}
\right) \equiv \left(
\begin{matrix}
H_{+} & 0 \\
0 & H_{-}
\end{matrix}
\right), \label{Anticomm}
\end{gather}
and $A^{\pm }$ are defined in terms of the derivative $\frac{{\rm d}}{{\rm d}x}$ and an arbitrary complex function~$w(r) $, called the \emph{superpotential}
\begin{gather}
A^{\pm }=\mp \frac{{\rm d}}{{\rm d}r}+w(r). \label{Superpot}
\end{gather}
Since the products $A^{+}A^{-}$ and $A^{-}A^{+}$ with $A^{\pm }$ defined in~(\ref{Superpot}) lead to
\begin{gather*}
A^{+}A^{-}=-\frac{{\rm d}^{2}}{{\rm d}r^{2}}+w^{2}-\frac{{\rm d}w}{{\rm d}r} ,\qquad A^{-}A^{+}=-\frac{{\rm d}^{2}}{{\rm d}r^{2}}+w^{2}+\frac{{\rm d}w}{{\rm d}r},
\end{gather*}
then, from (\ref{Anticomm}) it results natural to identify
\begin{gather*}
H_{\pm }=-\frac{{\rm d}^{2}}{{\rm d}r^{2}}+w^{2}\pm \frac{{\rm d}w}{{\rm d}r}, 
\end{gather*}
which leads directly to a definition of the so-called partner potentials $v_{\pm }$ given by
\begin{gather}
v_{\pm }\equiv w^{2}\pm \frac{{\rm d}w}{{\rm d}r}, \label{Partner Pots}
\end{gather}%
such that
\begin{gather*}
H_{\pm }=-\frac{{\rm d}^{2}}{{\rm d}r^{2}}+v_{\pm }. 
\end{gather*}
Each of the two equations in (\ref{Partner Pots}) define a Riccati differential equation for the superpotential~$w$, if $v_{\pm }$ are known. Let's recall that the superpotential can also be found from the zero-energy base state $\psi_0$, by computing $w=-\psi'_0/\psi_0$, where $\psi_0$ is a solution of the Schr\"{o}dinger equation with the $v_{-}$ potential
\begin{gather}\label{groundfunction} \psi''_0=v_{-}\psi_0, \qquad
\psi_0={\rm e}^{-\int w{\rm d}r} \end{gather} (see Witten~\cite{witten}). Riccati equations play a fundamental role in the study of integrability in SUSYQM. For a systematic study of this subject see references \cite{acthesis,acbook}.

\subsection*{Differential Galois theory}
Exact solutions of differential equation is a hard but important task in different disciplines. Sometimes numerical methods cannot be implemented in general, if the equation has free generic parameters. \emph{Differential Galois theory}, also known as \emph{Picard--Vessiot theory}, is a powerful theory to solve explicitly, in the case when it is possible, linear differential equations.

Analogous to the concept of field in classical Galois theory, there exists the concept \emph{dif\-fe\-rential field} in differential Galois theory, which is a field satisfying the differential Leibniz rules. Similarly, a~differential extension~$L$ of the differential field~$K$ means that~$K$ is a subfield of~$L$ preserving the differential Leibniz rules. In particular for a given linear differential with coefficients in~$K$, if $C_L=C_K$ (the field of constants of~$L$ is the same field of constants of~$K$) and~$L$ is generated over~$K$ by a fundamental set of solutions of such differential equation, then~$L$ is called the \emph{Picard--Vessiot extension of~$K$}. Recall that the field of constants of~$K$ is defined as $C_K:=\{k\in K\colon k'=0\}$, where $':={\rm d}/{\rm d}x$.

In the same way as we are interested in finding the roots of the polynomials
over a base field, usually $\mathbb{Q}$, using arithmetical and algebraic conditions, we would
like to have explicit solutions of differential equations over a
differential base field $K=\mathbb{C}(x)$, with field of constants $C_K=\mathbb{C}$, using elementary functions and
quadratures. The differential Galois theory considers more general differential fields, but for our purpose is enough to consider $\mathbb{C}(x)$.
Thus, the differential Galois group ($\mathrm{DGal}(L/K)$), as analogically as in the polynomial case, is the group of all \emph{differential automorphisms} that restricted to the base field coincide with the identity. Moreover if $\langle y_1,y_2,\ldots,y_n\rangle$ is a basis of solutions of
\begin{gather*}\frac{{\rm d}^ny}{{\rm d}x^n}+a_{n-1}\frac{{\rm d}^{n-1}y}{{\rm d}x^{n-1}}+\cdots+a_1\frac{{\rm d}y}{{\rm d}x}+a_0y=0,\qquad a_i\in\mathbb{C}(x),\end{gather*} then for each differential automorphism $\sigma\in\mathrm{DGal}(L/K)$ there exists a matrix $A_\sigma\in \mathrm{GL}(n,\mathbb{C})$ (i.e., $a_{ij}\in \mathbb{C}$, $1\leq i,j\leq n$ and $\det(A_\sigma)\neq 0$) such that
\begin{gather*}\sigma(\mathbf{Y})=A_\sigma\mathbf{Y},\qquad \mathbf{Y}=\begin{pmatrix}z_1\\z_2\\ \vdots\\ z_n
\end{pmatrix},\\ A_\sigma=\begin{pmatrix}\alpha_{11}&\alpha_{12}&\ldots& \alpha_{1n}\\
\alpha_{21}&\alpha_{22}&\ldots& \alpha_{2n}\\ \vdots & \vdots& \vdots&\vdots\\
\alpha_{n1}&\alpha_{n2}&\ldots& \alpha_{nn}\end{pmatrix},\qquad \mathrm{DGal}(L/K)\cong G\subset \mathrm{GL}(n,\mathbb{C}).\end{gather*}
In particular, $\mathrm{SL}(n,\mathbb{C})=\{A\in\mathrm{GL}(n,\mathbb{C})\colon \det(A)=1\}$. Due to $G=\{A_\sigma\colon \sigma\in \mathrm{DGal}(L/K)\}\subset \mathrm{GL}(n,\mathbb{C})$, we see that $\mathrm{DGal}(L/K)$ has a faithful representation as an algebraic group of matrices in where $G^0$ denotes the connected identity component of $G$ (the biggest algebraic connected subgroup of $G$). In this terminology, we say that a linear differential equation is \emph{integrable in the sense of differential Galois theory} whether the connected identity component of its differential Galois is a solvable group. Moreover, this definition of integrability leads to the obtaining of solutions in closed form if and only if $G^0$ is solvable, see~\cite{vasi} for full explanation and details. From now on, integrable in this paper means integrable in terms of differential Galois theory, see~\cite{Singer81}.

To accomplish our purposes, we are interested in second-order differential equations of the form
\begin{gather}\label{soldeq}
z''+az'+bz=0,\qquad a,b\in \mathbb{C}(x).
\end{gather}
Equation \eqref{soldeq} can be transformed into equations in the form
\begin{gather}\label{redsec}
 y''=ry,\qquad r={a^2\over 4}+{a'\over 2}-b,\qquad \text{and}\qquad z={\rm e}^{{-1\over 2}\int a {\rm d}x}y,
\end{gather}
see \cite{almp}. Jerald Kovacic developed in 1986 an algorithm to solve explicitly second-order differential equations with rational coefficients given in the form of equation~\eqref{redsec}, see~\cite{kov86}. In~\cite{dulo} another version of Kovacic's algorithm is presented, and it is applied to solve several second-order differential equations with special functions as solutions. The version of Kovacic's algorithm presented here corresponds to~\cite{acbl}, see also~\cite{acthesis, almp, acmowe, acsu}.

As mentioned, Kovacic's algorithm cannot be applied when the coefficients of the second-order differential equations are not rational functions. Therefore we need to transform such differential equations to apply Kovacic's algorithm. A possible solution to this problem was developed in~\cite{acthesis,acbook, acmowe}, the so-called Hamiltonian algebrization. However, we are interested in transformations that preserve the differential Galois
group (at least their connected identity component), in other words, the transformation must be either \emph{isogaloisian}, \emph{virtually isogaloisian} or \emph{strongly isogaloisian}, see~\cite{acthesis,acmowe}.

One important differential equation in this work is the Whittaker's differential equation, which is given by
\begin{gather}\label{whittaker} \partial_x^2y=\left(\frac14-{\kappa\over x}+{4\mu^2-1\over 4x^2}\right)y.\end{gather}
The Galoisian structure of this equation has been deeply studied in~\cite{Martinet-Ramis}, see also~\cite{dulo}. The following theorem provides the conditions of the integrability in the sense of differential Galois theory of equation~\eqref{whittaker}.

\begin{Theorem}[\cite{Martinet-Ramis}]\label{thmarram} The Whittaker's differential equation \eqref{whittaker} is integrable $($in the sense of differential Galois theory$)$ if and only if either, $\kappa+\mu\in\frac12+\mathbb{N}$, or $\kappa-\mu\in\frac12+\mathbb{N}$, or $-\kappa+\mu\in\frac12+\mathbb{N}$, or $-\kappa-\mu\in\frac12+\mathbb{N}$.
\end{Theorem}
The \textit{Bessel's equation} is a particular case of the confluent hypergeometric equation and is given by
\begin{gather}\label{bessel} \partial_x^2y+{1\over x}\partial_xy+{x^2-n^2\over x^2}y=0.\end{gather}
Under a suitable transformation, the reduced form of the Bessel's equation is a particular case of the Whittaker's equation. Thus, we can obtain the following well known result, see \cite[p.~417]{kol} and see also \cite{kov86,mo}:

\begin{Corollary}\label{corbessel} The Bessel's differential equation \eqref{bessel} is integrable $($solvable by quadratures$)$ if and only if $n\in \frac12+\mathbb{Z}$.
\end{Corollary}

\begin{Definition}[Hamiltonian change of variable, \cite{acbl}]\label{def2} A change of variable $z=z(x)$ is called \textit{Hamiltonian} if $(z(x),\partial_xz(x))$ is a solution curve of the autonomous classical one degree of freedom Hamiltonian system
\begin{gather*} \partial_xz=\partial_wH,\qquad \partial_xw=-\partial_zH \qquad \textrm{with}\quad H=H(z,w)={w^2\over 2}+V(z),\end{gather*}
for some $V\in K$.
\end{Definition}

\begin{Proposition}[Hamiltonian algebrization, \cite{acbl}]\label{pr2} The differential equation
\begin{gather*}\partial_x^2{r}=q(x)r\end{gather*}
is algebrizable through a Hamiltonian change of variable $z=z(x)$ if and only if there exist~$f$,~$\alpha$ such that
\begin{gather*}{\partial_z\alpha\over\alpha},\qquad {f\over \alpha}\in \mathbb{C}(z), \qquad \text{where} \qquad f(z(x))=q(x),\qquad \alpha(z)=2(H-V(z))=(\partial_xz)^2.\end{gather*}
Furthermore, the algebraic form of the equation $\partial_x^2{r}=q(x)r$ is
\begin{gather*}
\partial_z^2y+{1\over2}{\partial_z\alpha\over \alpha}\partial_zy-{f\over\alpha}y=0,\qquad r(x)=y(z(x)).
\end{gather*}
\end{Proposition}
Next, we follow the references \cite{acthesis, acbl,acmowe} to describe Kovacic's algorithm. Thus, to solve second-order differential equations with rational coefficients we use should Kovacic's algorithm, which is presented in Appendix~\ref{app1}.

\section{Main results}\label{section3}

In this section we present the main contributions of this paper. First we will show that for the usual $\nu =6$, $\delta =12$ Lennard-Jones potential, the Schr\"{o}dinger equation is non-integrable in the sense of differential Galois theory for any value of energy. In contrast for $\delta=10$ and $\nu=6$ we show the integrability, in the sense of differential Galois theory, as a special case of a general theorem for $\delta =2\nu -2$ with $\delta ,\nu \in \mathbb{N}$ (see Theorem~\ref{tprin} and its subsequent remark). From the physical point of view, the $10-6$ case is of the most remarkable importance. Since we preserve the physically grounded $-{1/r}^{6}$ term coming from dipole-dipole interactions and responsible of the van der Waals forces; but we replace never the less, the rather arbitrary ${1/r}^{12}$ term responsible for the repulsion of the particles in the many body system, and leading to a~non-integrable differential equation; with an equally arbitrary ${1/r}^{10}$ term, \textit{but leading to an integrable one}. We will dedicate the subsequent sections to show the advantages and physical interest of this special choice (see Fig.~\ref{Fig2} for a graphic comparison of both potentials).

\begin{figure}[t]\centering
\includegraphics[width=80mm]{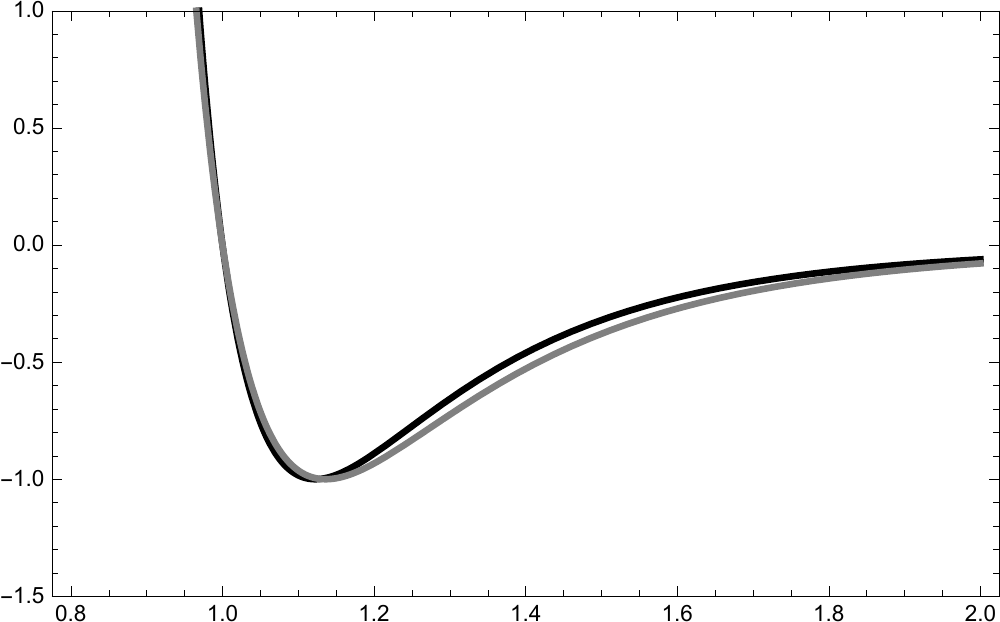}
\caption{$\frac{v(r)}{\epsilon}$ vs.~$\frac{r}{\sigma}$ plot for the rescaled $12-6$ (black) and $10-6$ (grey) Lennard-Jones potentials given in equation~\eqref{Gral LJ Resc} for the same molecular parameters~$\sigma$,~$\epsilon$ and $C=0$.}\label{Fig2}
\end{figure}
We start by considering the radial Schr\"{o}dinger equation (\ref{Con H}) with the generalized effective Lennard-Jones potential (\ref{Gral LJ Resc})
\begin{gather}
\mathcal{H}u_{k,l}(r) = \varepsilon _{k,l}u_{k,l}(r) ,\qquad
\mathcal{H} \equiv -\frac{{\rm d}^{2}}{{\rm d}r^{2}}+v_{\delta -\nu }^{\rm eff}(r)=-\frac{{\rm d}^{2}}{{\rm d}r^{2}}-{\frac{\bar{A}}{r^{\nu }}}+{\frac{\bar{B}}{r^{\delta }}}+\frac{\bar{C}}{r^{2}}, \nonumber \\
0 < \nu <\delta \in \mathbb{N\subset Z},\qquad \bar{A}>0,\qquad \bar{B}>0,\qquad \bar{C}\geq 0. \label{Con H 2}
\end{gather}
Setting $\mathbb{C}(r)$ as the differential field of equation (\ref{Con H 2}) with the derivative $\frac{{\rm d}}{{\rm d}r}$, we set also \smash{$\bar{A},\bar{B},\bar{C}\in \mathbb{C}$}.

\begin{Theorem}\label{theo12-6}Schr\"odinger equation with original $12-6$ Lennard-Jones effective potential is not integrable in the sense of differential Galois theory for any value of the energy and for all $A, B\in \mathbb{C}^*$, $C \in\mathbb{C}$.
\end{Theorem}
\begin{proof} Considering $\nu=6$ and $\delta=12$ in equation \eqref{Con H 2} we arrive to the Schr\"odinger equation with effective original $12-6$ Lennard-Jones potential. Now, applying the Hamiltonian change of variable $z=r^{2}$ over such Schr\"{o}dinger equation we arrive to the differential equation
\begin{gather*}
u''_{k,l} +\frac{1}{2z} u'_{k,l}+\left( \frac{A}{4z^{4}}-\frac{B}{4z^{7}}-\frac{C}{4z^{2}}
+\frac{\varepsilon _{k,l} }{4z}\right)u_{k,l}=0.\end{gather*} Now, the change of dependent variable
\begin{gather*}u_{k,l}={\Phi_{k,l} \over \sqrt[4]{z}}\end{gather*} leads to the differential equation
\begin{gather}
\Phi_{k,l}''=\left( {-4\varepsilon_{k,l}z^6+(4C-3)z^5-4Az^3+4B\over 16z^7} \right) \Phi_{k,l}. \label{Fi sec}
\end{gather}
After applying Kovacic's algorithm, see Appendix~\ref{app1}, we observe that equation \eqref{Fi sec} falls in case~4 for any $\varepsilon_{k,l} \in \mathbb{C}$ because there are not suitable conditions in step~1 for case~1 and case~3. The second step is not satisfied in case 2 because $D=\varnothing$ due to $E_{0}= \{7 \}$, $E_{\infty}= \{1,2 \}$ and there are not integers satisfying the condition for $D\neq \varnothing$. Thus we conclude that Schr\"odinger equation with original $12-6$ Lennard-Jones effective potential is not integrable in the sense of differential Galois theory for any value of the energy.
 \end{proof}

\subsection*{Supersymmetric quantum mechanics and the Lennard-Jones superpotential}
The implementation of Hamiltonian algebrization and Kovacic's algorithm reaches a considerable power in the realm of supersymmetric quantum mechanics. In fact the integrability of second-order linear equations like the radial Schr\"{o}dinger equation~(\ref{Con H 2}) subject of our study, via the Kovacic's algorithm is deeply related with the properties of the solutions of the associated Riccati equation in the supersymmetric extension of the theory~\cite{acthesis}. Taking this as a motivation, we go further in this section and propose a superpotential leading to the non-effective part~(\ref{LJ Resc 1}) of~(\ref{Gral LJ Resc}) ($C=0$) as one of two partner potentials. If we denote the superpotential in one dimension as~$w(r) $ the corresponding partner potentials are given by equation~(\ref{Partner Pots}) (see \cite{acthesis,Gango,witten}, among others)
\begin{gather}
v_{\pm }(r) \equiv w^{2}(r) \pm \frac{{\rm d}w}{{\rm d}r}\label{Def Pots}
\end{gather}
corresponding for each case to a Riccati equation for $w$. Knowing that $0<\nu <\delta $ we identify terms in~(\ref{Def Pots}) with terms in~(\ref{LJ Resc 1}) as follows
\begin{gather*}
w^{2}(r) ={\frac{\bar{B}}{r^{\delta }},\qquad \frac{{\rm d}w}{{\rm d}r}=\frac{\bar{A}}{r^{\nu }}} 
\end{gather*}
as a consequence we have
\begin{gather*}
w(r) =\pm {\frac{\sqrt{\bar{B}}}{r^{\frac{\delta }{2}}},\qquad w(r) =-\frac{\bar{A}}{(\nu -1) r^{\nu -1}}+C}.
\end{gather*}
A simple choice for $w(r) $ is given by%
\begin{gather}
w(r) \equiv -{\frac{\sqrt{\bar{B}}}{r^{\frac{\delta }{2}}}},\label{simple choice}
\end{gather}%
where the following identities should hold
\begin{gather}
\sqrt{\bar{B}}=\bar{A}/(\nu -1),\qquad \delta =2(\nu -1),\qquad C=0. \label{Conds Superpot}
\end{gather}
As a result we identify $v_{\delta -\nu }(r)$ in~(\ref{LJ Resc 1}) with $v_{-}$ and we have from~(\ref{Def Pots}), the following expressions for the partner potentials
\begin{gather}
v_{-}={\frac{\bar{B}}{r^{\delta }}-\frac{\bar{A}}{r^{\nu }}=v_{\delta -\nu}(r)},\qquad v_{+}={\frac{\bar{B}}{r^{\delta }}+\frac{\bar{A}}{r^{\nu }}}.\label{P Pots D-Nu}
\end{gather}
According to equation~\eqref{groundfunction}, the corresponding wave function for the zero energy level using~$v_{-}$ is given by
\begin{gather}\label{gsfunction}\psi_{0}(r) ={\rm e}^{-\frac{2\sqrt{\bar{B}}}{(\delta-2)r^{\frac{\delta-2}{2}}}}={\rm e}^{-\frac{A}{(\nu-1)(\nu -2)r^{\nu-2}}+Cr}.\end{gather}

\begin{figure}[t]\centering
\includegraphics[width=60mm]{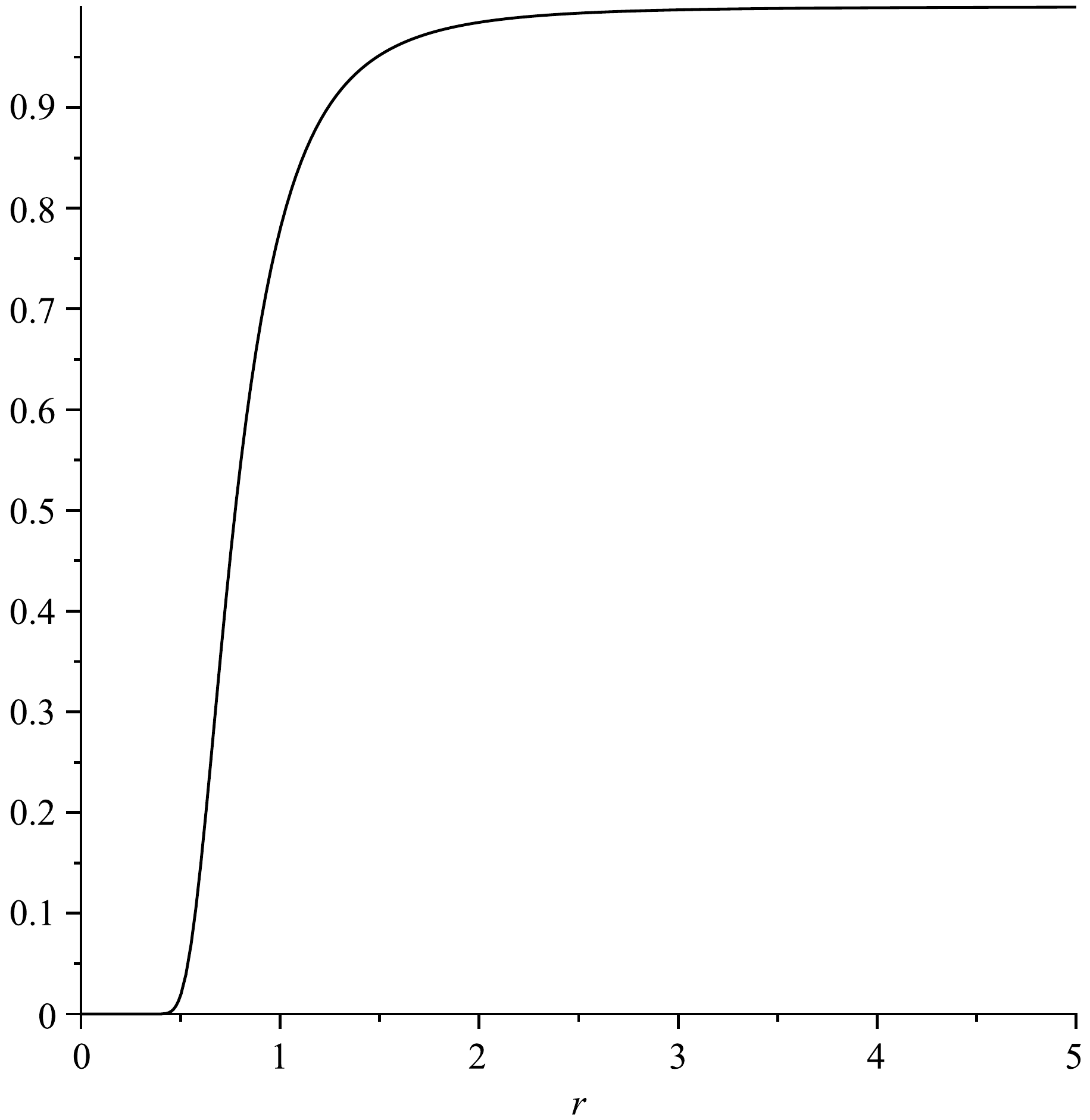}
\caption{Wave function for $v_{10-6}$ with zero energy, $\bar{A}=5$, $\bar{B}=1$ and $C=0$.}\label{Fig3}\end{figure}

An example of wave function for this potential is given in Fig.~\ref{Fig3}. Summarizing we conclude that expressions in (\ref{Conds Superpot}) set conditions for the existence of a~superpotential in the form of~(\ref{simple choice}), and a supersymmetric extension to equation~(\ref{Con H 2}) with partner potentials in~(\ref{P Pots D-Nu}). The case for $\delta =2(\nu -1) $ thus appears in a natural way, as a simple condition for defining a supersymmetric model. The property of integrability for zero energy of this case to be proven in Theorem~\ref{tprin}, makes it an appealing model to further explore the relation between supersymmetry and integrability already studied in~\cite{acthesis}.

 \subsection*{The $\boldsymbol{10-6}$ Lennard-Jones superpotential and the De Boer parameter}
As mentioned before the case for $\nu =6$, $\delta =10$ is of particular relevance from physical grounds. One of the aims of this work is to explore the analytical advantages of $v_{10-6}$ in contrast to $v_{12-6}$. The $10-6$ potential in terms of the molecular parameters $\sigma $ and $\epsilon $ is given by
\begin{gather}
V_{10-6}(r) =\alpha \epsilon \left[\left(\frac{\sigma}{r}\right) ^{10 }-\left( \frac{\sigma}{r}\right)^{6}\right], \label{n-6}
\end{gather}%
where $\alpha$ is chosen so that $\epsilon$ is the minimum energy (the well depth) and $\sigma$, as mentioned before, is the value where $V_{10-6}$ vanishes. As a result we have for this case\footnote{Notice in (\ref{12-6 Sigma-Eps}) that $\alpha\equiv 4$ for the $12-6$ potential.} $\alpha \equiv (25/6) \sqrt{5/3}$. Thus the rescaled $10-6$ Lennard-Jones potential reads%
\begin{gather}
v_{10-6}(r)\equiv \left( \frac{2\mu }{\hbar ^{2}}\right) V_{10-6}(r)= \left( \frac{2\mu \alpha \epsilon}{\hbar ^{2}} \right) \left[ \left( \frac{\sigma}{r}\right)^{10 }-\left(\frac{\sigma}{r}\right)^{6}
\right] \label{Rescaled n-6}
\end{gather}
or in the equivalent $A-B$ form, we have the following
\begin{gather}
v_{10-6}(r) =\frac{\bar{B}}{r^{10}}-\frac{\bar{A}}{r^{6}} \qquad \text{with}\qquad \bar{A}\equiv \frac{2\mu \alpha\epsilon\left(\sigma \right) ^{6}}{ \hbar ^{2}},\qquad
\text{and}\qquad \bar{B}\equiv \frac{2\mu \alpha \epsilon (\sigma) ^{10 }}{\hbar ^{2}}. \label{V N-6 A-B}
\end{gather}%
Clearly $v_{10-6}$ fulfills the condition in~(\ref{Conds Superpot}) for $\nu=6$; as a result the rescaled superpotential for $\delta =10$ in (\ref {simple choice}) takes the form
\begin{gather*}
w_{10-6}(r) =-{\frac{\sqrt{\bar{B}}}{r^{5}}}, 
\end{gather*}
where $\bar{A}$ $=5\sqrt{\bar{B}}$ as we easily check from (\ref{Conds Superpot}). Equivalently
\begin{gather}
w_{10-6}(r) =- \frac{\bar{A}}{5r^{5}}= - \frac{2\mu \alpha\epsilon (\sigma)^{6}}{5\hbar ^{2}r^{5}}= - \frac{5\mu \sqrt{5/3}\epsilon (\sigma) ^{6}}{3\hbar^{2}r^{5}}. \label{Superpot 2_1}
\end{gather}%
The condition $\bar{A}$ $=5\sqrt{\bar{B}}$ can be written in the following suggestive dimensionless form
\begin{gather}
\frac{\hbar ^{2}}{\mu \epsilon (\sigma) ^{2}}=\frac{1 }{3}\sqrt{\frac{5}{3}}\approx 0.4303, \label{Susy cond cero}
\end{gather}
where we have applied definitions in~(\ref{V N-6 A-B}) and $\alpha \equiv (25/6) \sqrt{5/3}$. We will call it from now on the supersymmetric condition (for short SUSY condition) for the $10-6$ Lennard-Jones potential. In terms of the so-called De Boer parameter $\Lambda \equiv \hbar /(\sigma \sqrt{\mu \epsilon })$, which gives a degree of the quantum character of the system~\cite{DeBoer48}, we have $\Lambda ^{2}\approx \allowbreak 0.4303$ or similarly $\Lambda \approx 0.6559$. An important remark at this point is that the SUSY condition given in the form $\bar{A}$ $=5\sqrt{\bar{B}}$ will appear again in Theorems~\ref{theodk} and~\ref{tprin}, in the context of the Martinet--Ramis theorem, that is, Theorem~\ref{thmarram}.

As a summary of this section, we conclude that the fulfillment of condition~(\ref{Susy cond cero}), guarantees not only the solvability of the Schr\"{o}dinger equation~(\ref{Con H l0}) with the potential~$v_{10-6}$ in~(\ref{Rescaled n-6}) (through the Martinet--Ramis theorem, as we will see) but also the existence of a superpotential given by expression~(\ref{Superpot 2_1}), which correspondingly leads to $v_{10-6}$ as one of the partner poten\-tials~$v_{\pm }(r)$ defined through the Riccati equations in~(\ref{Def Pots}). The supersymmetric model thus formulated considers a~specific version of the radial Schr\"{o}dinger equation~(\ref{Hl0 Non}) or equivalently the rescaled form~(\ref{Con H l0}), where we set in both equations $l=0$ for the angular momentum, and~$V_{10-6}$ and~$v_{10-6}$ are given in~(\ref{n-6}) and~(\ref{Rescaled n-6}). We will start in the next section, a physical analysis of the model, in the light of the De Boer principle of corresponding states.

\subsection*{The low temperature behavior of the $\boldsymbol{10-6}$ Lennard-Jones gas}

Recalling the discussion in the previous section, about the dimensionless representation~(\ref{Dim Less}) of the Schr\"{o}dinger equation we start by noticing that the $10-6$ potential~(\ref{n-6}) fulfills the condition $V(r)=\epsilon f(r/\sigma)$ if we identify~$f(r/\sigma)$ with $\alpha \big[ \big( \frac{\sigma}{r}\big) ^{10}-\big( \frac{\sigma}{r}\big)^{6}\big]$; as a result we have
\begin{gather}
\tilde{V}_{10-6}(r) \equiv \frac{V_{10-6}(r)}{\epsilon} = ( 25/6 )\sqrt{5/3}\left[ \left( \frac{\sigma}{r}\right) ^{10}-\left( \frac{\sigma}{r}\right)^{6}\right] \nonumber \\
\hphantom{\tilde{V}_{10-6}(r)}{} = (25/6)\sqrt{5/3}\left[ \left( \frac{1}{\tilde{r}}\right) ^{10}-\left( \frac{1}{\tilde{r}}\right) ^{6}\right],\label{Dim Less 10-6 Pot}
\end{gather}
where we have used the definitions in (\ref{De Boer Trafo}) $\tilde{r}\equiv \frac{r}{\sigma}$ and $\tilde{V}\equiv \frac{V}{\epsilon}$. The Schr\"{o}dinger equation takes thus the form in (\ref{Dim Less}):
\begin{gather}
\left[ -\frac{\Lambda ^{2}}{2}\frac{{\rm d}^{2}}{{\rm d}\tilde{r}^{2}}+\tilde{V}_{10-6}(r)\right] u(r) =\tilde{E}u(r). \label{Dim Less 10-6}
\end{gather}%
As mentioned in Section~\ref{section2}, the De Boer principle of corresponding states tells us that two different systems with equal value of $\Lambda$ have identical thermodynamical properties~\cite{DeBoer48}. In this sense the SUSY condition $\Lambda^{2}=\hbar^{2}/\big[ (\sigma)^{2}\mu \epsilon\big] = (1/3 )\sqrt{5/3}\approx 0.4303$ given in~(\ref{Susy cond cero}); and representing a definite set of combinations of values of the parameters $\sigma$, $\epsilon$, and $\mu $; that accounts for $\Lambda ^{2}=(1/3)\sqrt{5/3}$; is defining through the principle of corresponding states, a specific set of physical systems with equivalent thermodynamical properties.
These systems have the special feature of being described by a Supersymmetric potential of the form~(\ref{Superpot 2_1}) leading to~(\ref{Dim Less 10-6}) with the potential~(\ref{Dim Less 10-6 Pot}) as the partner potential
$V_{-}(r)\equiv \big(\frac{\hbar^{2}}{2\mu}\big)v_{-}(r)$ in (\ref{P Pots D-Nu}) with $\nu=6$.

We have found after a brief review of the literature, a significant coincidence between the specific value for $\Lambda^{2}\approx 0.4303$ and the value of $\Lambda^{2}=0.456$ reported by Miller, Nosanow and Parish~\cite{Miller} for a second-order liquid to gas phase transition of a Bose--Einstein condensate at zero temperature. Since their calculation is an approximate one, made in the framework of the variational method; it is a worthy task (to be done elsewhere) to investigate the advantages of our exact approach to the calculation of properties of such many-body systems at low temperatures in the context of the quantum extension to the principle of corresponding states.

\looseness=-1 In Fig.~\ref{Fig4} we see a plot of the second virial coefficient calculated numerically for both the $12-6$ and $10-6$ potentials from the integral definition in~(\ref{Virial B}). Relying on $B_2(r)$ as a~quantity that gives information of the microscopic pair-potential (with the previously mentioned limitations) we see an asymptotic closeness of both functions for low temperatures, that hints for the reliabi\-li\-ty of our supersymetric model with the $10-6$ Lennard-Jones potential, near the absolute zero.
 \begin{figure}[t]\centering
\includegraphics[width=80mm]{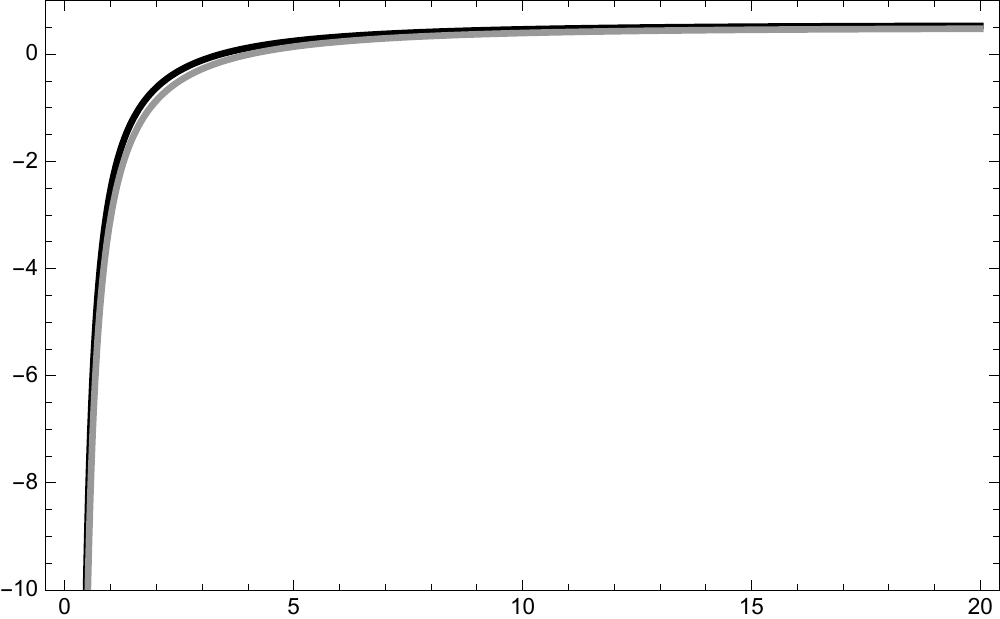}
\caption{$B_2(T)$ vs.\ $T/\epsilon$ plot for the $12-6$ (black) and $10-6$ (grey) Lennard-Jones potentials for the same molecular parameters~$\sigma$ and~$\epsilon$. The closeness of both functions for low temperatures near to absolute zero is a hint of the reliability of the $10-6$ potential in that region.}\label{Fig4}
\end{figure}

\vspace{-2mm}

\subsection*{Integrability of the $\boldsymbol{10-6}$ Lennard-Jones potential and its generalization}
The following result is valid for any potential $v(r)$ belonging to a differential field.
\begin{Theorem}\label{theodk} Consider $v(r)$ belonging to a differential field $K$, then the following statements hold
\begin{itemize}\itemsep=0pt
\item The only one change of dependent variable that allows to transform the radial equation of the Schr\"odinger equation into the Schr\"odinger equation with effective potential is $\varphi\colon u\mapsto \varphi(u)=ru$, where $u$ is the solution for the radial equation and $\varphi(u)$ is the solution for the Schr\"odinger equation with effective potential.
\item The differential Galois groups of the radial equation and Schr\"odinger equation with effective potential are subgroups of $\mathrm{SL}(2,\mathbb{C}) $.
\item The transformation $\varphi$ is strongly isogaloisian.
\end{itemize}
\end{Theorem}
\begin{proof}We proceed according to each item.
\begin{itemize}\itemsep=-1.5pt
\item Applying the transformation given in equation \eqref{soldeq} and equation \eqref{redsec} we obtain it because $2/r$ is the coefficient of the first derivative of the radial equation after separation of variables. Thus, applying the change of variable $\varphi\colon u\mapsto \varphi(u)=ru$ we arrive to the Schr\"odinger equation with effective potential.
\item The Wronskian of two independent solutions of the Schr\"odinger equation with effective potential is constant and constants are in the base field. Similarly, the Wronskian of two independent solutions of the radial equation belongs to the base field. Therefore, automorphisms over such solutions acts by multiplication of matrices belonging to $\mathrm{SL}(2,\mathbb{C})$, that is, $\sigma(U)=A_\sigma U$, $\sigma(\varphi(U))=A_\sigma \varphi(U)$ and $\det (A_\sigma)=1$. Thus, $A_\sigma\in \mathrm{SL}(2,\mathbb{C})$.
\item Applying the differential automorphism $\sigma$ over $\varphi(u)$ we observe that $\sigma(\varphi(u))=\sigma(r)\sigma(u)=r\sigma(u)$, which implies that differential Galois group only depends on the solutions $u$ because~$r$ is in the base field and the differential Galois group will be the same with the same base field. Thus, the transformation $\varphi$ is strongly isogaloisian.
\end{itemize}
Thus we conclude the proof.
\end{proof}

The following result corresponds to the integrability conditions for the $10-6$ L-J potential and its generalization.
\begin{Theorem}\label{tprin}The Schr\"{o}dinger equation with $(2\nu-2)-\nu$ L-J potential, given in equation~\eqref{Con H 2}, is integrable for zero energy in the sense of differential Galois theory if and only if
\begin{gather*}
A=\pm\sqrt{B}\big(\pm\sqrt{1+4C}+\nu-2+2mk-4m\big),\qquad m\in \mathbb{Z}.
\end{gather*}
\end{Theorem}
\begin{proof}The Schr\"{o}dinger equation given in equation \eqref{Con H 2}, with zero energy, is transformed into the Whittaker's differential equation \eqref{whittaker} through the change of variables
\begin{gather*}u_{k,l}=\sqrt{r^{\nu-1}}\phi_{k,l},\qquad r=\sqrt[\nu-2]{2\sqrt{B}\over (\nu-2)z}\end{gather*} with parameters \begin{gather*}\kappa={A\over \sqrt{B}(2\nu-4)}\qquad \text{and}\qquad \mu={\sqrt{1+4C}\over 2\nu-4}.\end{gather*} Applying Martinet--Ramis theorem we have
\begin{gather*}\pm{A\over \sqrt{B}(2\nu-4)}\pm {\sqrt{1+4C}\over 2\nu-4}\in \mathbb{Z}+\frac12.\end{gather*} Assuming $m\in\mathbb{Z}$ we obtain \begin{gather*}A=\pm\sqrt{B}\big(\pm\sqrt{1+4C}+\nu-2+2m\nu-4m\big),\qquad m\in \mathbb{Z},\end{gather*}
which is the integrability condition for the Schr\"{o}dinger equation with $(2\nu-2)-\nu$ L-J potential and its wave function corresponds to equation~\eqref{gsfunction} with $\delta=2\nu-2$.
\end{proof}

\begin{Remark}We observe that Theorem~\ref{tprin} refers to the integrability in the sense of differential Galois theory, which is not related with square integrable wave functions. Another key point is that we are not considering energies different than zero and excited states, this is an open problem for this generalized potential. In particular, the theorem includes the $10-6$ L-J potential, for $C=0$ and $\nu=6$. Therefore the Schr\"odinger equation with $L-J$ 10-6 is integrable for zero energy when $A=\pm\sqrt{B}(8m+4\pm 1)$, while energies different than zero and excited states were not considered in this paper. Moreover,
for zero energy and $m=-1$ we recover the integrability condition obtained through SUSYQM for this potential, i.e., the Schr\"{o}dinger equation with $10-6$ L-J potential is also integrable in the sense of differential Galois theory for $A\in\big\{\pm 3\sqrt{B},\pm 5\sqrt{B} \big\}$.
\end{Remark}

\section{Final remarks and open questions}\label{section4}

In this paper we have shown that there exist no explicit solutions of the radial Schr\"{o}dinger equation with the usual $12-6$ Lennard-Jones potential for any value of the energy. We have proposed an alternative supersymmetric model with a $10-6$, $v_{-}$ partner potential, that preserves the~$-1/r^6$ van der Waals attraction. We have found through the De Boer principle of corresponding states, initial hints that this model could represent a low temperature system determined by a $\Lambda^2 \approx 0.4303$ value of the 2nd power of the De Boer parameter. We have studied possible generalizations of the Lennard-Jones potential, where the Schr\"{o}dinger equation is integrable in the sense of differential Galois theory.

Further work can be developed looking for similar theorems of integrability in the sense of differential Galois theory for $E\neq 0$ and excited states, for the $10-6$ potential and other generalizations. Relations between square integrable wave functions and solutions in closed form of SE for generalizations of L-J potentials should be explored in further works too.

We hope that this paper can be the starting point of further works involving SUSYQM, DGT and statistical mechanics, which are not easy topics. Although we tried to write a readable preliminaries about these topics, we know that it was not enough and the reader should complement with references suggested by us, otherwise this paper could be a large paper, which was not the target.

\appendix
\section{Kovacic algorithm}\label{app1}
The version of Kovacic's algorithm presented in this appendix is based in the improved version given in~\cite{acbl}. There are four cases in Kovacic's algorithm. Only for cases 1, 2 and 3 we can solve the differential equation, but for the case~4 the differential equation is not integrable. It is possible that Kovacic's algorithm can provide us only one solution~($\zeta_1$), so that we can obtain the second solution ($\zeta_2$) through
\begin{gather*}
\zeta_2=\zeta_1\int\frac{{\rm d}x}{\zeta_1^2}.
\end{gather*}

{\bf Notations.} For the differential equation given by
\begin{gather*}\partial_x^2\zeta=r\zeta,\qquad r={s\over t},\qquad s,t\in \mathbb{C}[x],\end{gather*}
we use the following notations:
\begin{enumerate}\itemsep=0pt
\item[1)] denote by $\Gamma'$ be the set of (finite) poles of $r$, $\Gamma^{\prime}= \{ c\in\mathbb{C}\colon t(c)=0 \}$,
\item[2)] denote by $\Gamma=\Gamma^{\prime}\cup\{\infty\}$,
\item[3)] by the order of $r$ at $c\in \Gamma'$, $\circ(r_c)$, we mean the multiplicity of $c$ as a~pole of $r$,
\item[4)] by the order of $r$ at $\infty$, $\circ (r_{\infty})$, we mean the order of $\infty$ as a zero of~$r$. That is $\circ ( r_{\infty } )=\deg (t)-\deg (s)$.
\end{enumerate}

\subsection*{The four cases}

{\bf Case 1.} In this case $[ \sqrt{r}] _{c}$ and $[ \sqrt{r}] _{\infty}$ means the Laurent series of $\sqrt{r}$ at $c$ and the Laurent series of $\sqrt{r}$ at $\infty$ respectively. Furthermore, we define $\varepsilon(p)$ as follows: if $p\in\Gamma$, then $\varepsilon( p) \in\{+,-\}$. Finally, the complex numbers $\alpha_{c}^{+}$, $\alpha_{c}^{-}$, $\alpha_{\infty}^{+}$, $\alpha_{\infty}^{-}$ will be defined in the first step. If the differential equation has no poles it only can fall in this case.

{\bf Step 1.} Search for each $c \in \Gamma'$ and for $\infty$ the corresponding situation as follows:
\begin{enumerate}\itemsep=0pt
\item[$(c_{0})$] If $\circ(r_{c}) =0$, then \begin{gather*} [ \sqrt {r} ] _{c}=0,\qquad\alpha_{c}^{\pm}=0.\end{gather*}

\item[$(c_{1})$] If $\circ(r_{c}) =1$, then \begin{gather*} [ \sqrt {r} ] _{c}=0,\qquad\alpha_{c}^{\pm}=1.\end{gather*}

\item[$(c_{2})$] If $\circ(r_{c}) =2$, and \begin{gather*}r= \cdots
+ b(x-c)^{-2}+\cdots,\qquad \textrm{then}\qquad [ \sqrt {r} ]_{c}=0,\qquad \alpha_{c}^{\pm}=\frac{1\pm\sqrt{1+4b}}{2}.\end{gather*}

\item[$(c_{3})$] If $\circ(r_{c}) =2v\geq4$, and \begin{gather*}r=
\big(a(x-c) ^{-v}+\dots +d(x-c)^{-2}\big)^{2}+b(x-c)^{-(v+1)}+\cdots,\qquad \textrm{then}\\ [\sqrt {r}] _{c}=a(x-c) ^{-v}+\dots +d(x-c) ^{-2},\qquad\alpha_{c}^{\pm}=\frac{1}{2}\left(
\pm\frac{b}{a}+v\right).\end{gather*}

\item[$(\infty_{1})$] If $\circ ( r_{\infty} ) >2$, then
\begin{gather*}[\sqrt{r}] _{\infty}=0,\qquad\alpha_{\infty}^{+}=0,\qquad\alpha_{\infty}^{-}=1.\end{gather*}

\item[$(\infty_{2})$] If $\circ(r_{\infty}) =2$, and
$r= \cdots + bx^{2}+\cdots$, then \begin{gather*} [\sqrt{r}] _{\infty}=0,\qquad\alpha_{\infty}^{\pm}=\frac{1\pm\sqrt{1+4b}}{2}.\end{gather*}

\item[$(\infty_{3})$] If $\circ ( r_{\infty} ) =-2v\leq0$, and
\begin{gather*}r=\big( ax^{v}+\dots +d\big)^{2}+ bx^{v-1}+\cdots,\qquad \textrm{then}\\
 [\sqrt{r}] _{\infty}=ax^{v}+\dots +d,\qquad \text{and}\qquad \alpha_{\infty}^{\pm }=\frac{1}{2}\left(\pm\frac{b}{a}-v\right).\end{gather*}
\end{enumerate}

{\bf Step 2.} Find $D\neq\varnothing$ defined by
\begin{gather*}D=\left\{n\in\mathbb{Z}_{+}\colon n=\alpha_{\infty}^{\varepsilon
(\infty)}- \sum\limits_{c\in\Gamma^{\prime}}
\alpha_{c}^{\varepsilon(c)},\, \forall\, ( \varepsilon (p) ) _{p\in\Gamma}\right\} .\end{gather*} If $D=\varnothing$, then we should start with the case~2. Now, if $\mathrm{Card}(D)>0$, then for each $n\in D$ we search~$\omega$ $\in\mathbb{C}(x)$ such that
\begin{gather*}\omega=\varepsilon(\infty) [\sqrt{r}] _{\infty}+\sum\limits_{c\in\Gamma^{\prime}}
\big( \varepsilon(c) [\sqrt{r}] _{c}+{\alpha_{c}^{\varepsilon(c)}}{(x-c)^{-1}}\big).\end{gather*}

{\bf Step 3}. For each $n\in D$, search for a monic polynomial $P_n$ of degree $n$ with
\begin{gather*}
\partial_x^2P_n + 2\omega \partial_xP_n + \big(\partial_x\omega + \omega^2 - r\big) P_n = 0.
\end{gather*}
If success is achieved then $\zeta_1=P_n {\rm e}^{\int\omega}$ is a~solution of the differential equation. Else, case~1 cannot hold.

{\bf Case 2.} Search for each $c \in \Gamma'$ and for $\infty$ the corresponding situation as follows:

{\bf Step 1.} Search for each $c\in\Gamma^{\prime}$ and $\infty$ the sets $E_{c}\neq\varnothing$ and $E_{\infty}\neq\varnothing$. For each $c\in\Gamma^{\prime}$ and for $\infty$ we define $E_{c}\subset\mathbb{Z}$ and $E_{\infty}\subset\mathbb{Z}$ as follows:
\begin{enumerate}\itemsep=0pt
\item[($c_1$)] If $\circ(r_{c})=1$, then $E_{c}=\{4\}$.

\item[($c_2$)] If $\circ(r_{c}) =2$, and $r= \cdots + b(x-c)^{-2}+\cdots$, then $E_{c}= \big\{2+k\sqrt{1+4b}\colon k=0,\pm2\big\}$.

\item[($c_3$)] If $\circ(r_{c}) =v>2$, then $E_{c}=\{v\}$.

\item[$(\infty_{1})$] If $\circ(r_{\infty}) >2$, then $E_{\infty }=\{0,2,4\}$.

\item[$(\infty_{2})$] If $\circ(r_{\infty}) =2$, and $r= \cdots + bx^{2}+\cdots$, then $E_{\infty }= \big\{2+k\sqrt{1+4b}\colon k=0,\pm2 \big\}$.

\item[$(\infty_{3})$] If $\circ(r_{\infty}) =v<2$, then $E_{\infty }=\{v\}$.
\end{enumerate}

{\bf Step 2.} Find $D\neq\varnothing$ defined by
\begin{gather*}D=\left\{
n\in\mathbb{Z}_{+}\colon n=\frac{1}{2}\left( e_{\infty}- \sum\limits_{c\in\Gamma^{\prime}} e_{c}\right),\, \forall\, e_{p}\in E_{p},\, p\in\Gamma\right\}.\end{gather*} If
$D=\varnothing$, then we should start the case~3. Now, if $\mathrm{Card}(D)>0$, then for each $n\in D$ we search a rational function $\theta$ defined by
\begin{gather*}\theta=\frac{1}{2} \sum\limits_{c\in\Gamma^{\prime}} \frac{e_{c}}{x-c}.\end{gather*}

{\bf Step 3.} For each $n\in D$, search a monic polynomial $P_n$ of degree $n$, such that
\begin{gather*}
\partial_x^3P_n+3\theta \partial_x^2P_n+\big(3\partial_x\theta+3\theta ^{2}-4r\big)\partial_xP_n+\big( \partial_x^2\theta+3\theta\partial_x\theta +\theta^{3}-4r\theta-2\partial_xr\big)P_n=0.
\end{gather*}
If $P_n$ does not exist, then case 2 cannot hold. If such a polynomial is found, set $\phi = \theta + \partial_xP_n/P_n$ and let $\omega$ be a solution of{\samepage
\begin{gather*}\omega^2 + \phi \omega + {1\over2}\big(\partial_x\phi + \phi^2 -2r\big)= 0.\end{gather*}
Then $\zeta_1 = {\rm e}^{\int\omega}$ is a solution of the differential equation.}

{\bf Case 3.} Search for each $c \in \Gamma'$ and for $\infty$ the corresponding situation as follows:

{\bf Step 1.} Search for each $c\in\Gamma^{\prime}$ and $\infty$ the sets $E_{c}\neq\varnothing$ and $E_{\infty}\neq\varnothing$. For each $c\in\Gamma^{\prime}$ and for $\infty$ we define $E_{c}\subset\mathbb{Z}$ and $E_{\infty}\subset\mathbb{Z}$ as follows:
\begin{enumerate}\itemsep=0pt
\item[$(c_{1})$] If $\circ(r_{c}) =1$, then $E_{c}=\{12\}$.

\item[$(c_{2})$] If $\circ(r_{c}) =2$, and $r= \cdots +b(x-c)^{-2}+\cdots$, then
\begin{gather*}
E_{c}= \big\{ 6+k\sqrt{1+4b}\colon k=0,\pm1,\pm2,\pm3,\pm4,\pm5,\pm6 \big\}.
\end{gather*}

\item[$(\infty)$] If $\circ(r_{\infty}) =v\geq2$, and $r=
\cdots + bx^{2}+\cdots$, then \begin{gather*}E_{\infty }= \left\{ 6+{12k\over m}\sqrt{1+4b}\colon k=0,\pm1,\pm2,\pm3,\pm4,\pm5,\pm6\right\},\qquad
m\in\{4,6,12\}.\end{gather*}
\end{enumerate}

{\bf Step 2.} Find $D\neq\varnothing$ defined by
\begin{gather*}D=\left\{ n\in\mathbb{Z}_{+}\colon n=\frac{m}{12}\left(e_{\infty}- \sum\limits_{c\in\Gamma^{\prime}}e_{c}\right) ,\, \forall\, e_{p}\in E_{p},\, p\in\Gamma\right\}.\end{gather*}
In this case we start with $m=4$ to obtain the solution, afterwards $m=6$ and finally $m=12$. If $D=\varnothing$, then the differential equation is not integrable because it falls in the case~4. Now, if $\mathrm{Card}(D)>0$, then for each $n\in D$ with its respective~$m$, search a rational function
\begin{gather*}\theta={m\over 12} \sum\limits_{c\in\Gamma^{\prime}} \frac{e_{c}}{x-c}\end{gather*}
and a polynomial $S$ defined as \begin{gather*}S= \prod\limits_{c\in\Gamma^{\prime}} (x-c).\end{gather*}

{\bf Step 3}. Search for each $n\in D$, with its respective $m$, a~monic polynomial $P_n=P$ of degree~$n$, such that its coefficients can be determined recursively by
\begin{gather*} P_{-1}=0,\qquad P_{m}=-P,\\ P_{i-1}=-S\partial_xP_{i}- ( (m-i) \partial_xS-S\theta) P_{i}-(m-i) (i+1) S^{2}rP_{i+1},\end{gather*} where $i\in\{0,1,\dots,m-1,m\}$. If~$P$ does not exist, then the differential equation is not integrable because it falls in case~4. Now, if $P$ exists search~$\omega$ such that \begin{gather*} \sum\limits_{i=0}^{m}
\frac{S^{i}P}{(m-i) !}\omega^{i}=0,\end{gather*} then a solution
of the differential equation is given by \begin{gather*}\zeta={\rm e}^{\int \omega},\end{gather*} where $\omega$ is solution of the previous polynomial of degree~$m$.

\subsection*{Acknowledgements}
P.A.-H.~thanks to Universidad Sim\'on Bol\'{\i}var, Research Project \emph{M\'etodos Algebraicos y Combinatorios en Sistemas Din\'amicos y F\'{\i}sica Matem\'atica}. He also acknowledges and thanks the support of COLCIENCIAS through grant numbers FP44842-013-2018 of the Fondo Nacional de Financiamiento para la Ciencia, la Tecnolog\'{\i}a y la Innovaci\'on. E.T.~wishes to thank the German Service of Academic Exchange (DAAD) for financial support, and Professor M.~Reuter at the Institute of Physics in Uni-Mainz for stimulating discussions about this work.
Finally, the authors thank to the anonymous referees for their valuable comments and suggestions.


\pdfbookmark[1]{References}{ref}
\LastPageEnding

\end{document}